\documentclass[sigconf,review=false,balance=true,screen,nonacm]{acmart}
\usepackage{xcolor}
\usepackage{hyperref}

\newcommand{\cB}{\mathcal{B}}
\newcommand{\cR}{\mathcal{R}}
\newcommand{\cP}{\mathcal{P}}
\newcommand{\cC}{\mathcal{C}}
\newcommand{\cF}{\mathcal{F}}

\newcommand{\qbinom}[3][q]{\genfrac[]{0pt}{}{#2}{#3}_{#1}}

\newcommand{\NN}{\mathbb{N}}
\newcommand{\ZZ}{\mathbb{Z}}
\newcommand{\QQ}{\mathbb{Q}}
\newcommand{\KK}{\mathbb{K}}

\DeclareMathOperator{\lcm}{lcm}

\numberwithin{equation}{section}
\theoremstyle{plain}
\newtheorem{theorem}{Theorem}[section]
\newtheorem{lemma}[theorem]{Lemma}
\newtheorem{corollary}[theorem]{Corollary}

\begin{document}

\title{Factorial Basis Method for \texorpdfstring{$q$}{q}-Series Applications}

\author{Antonio Jiménez-Pastor}
\orcid{0000-0002-6096-0623}
\affiliation{%
  \institution{Aalborg University}
  \department{Distributed, Embedded and Intelligent System}
  \city{Aalborg}
  \country{Denmark}
}
\email{ajpa@cs.aau.dk}
\author{Ali K. Uncu}
\orcid{0000-0001-5631-6424}
\affiliation{%
  \institution{RICAM, Austrian Academy of Sciences\\ \& University of Bath}
  \city{Linz}
  \country{Austria \& United Kingdom}
}
\email{akuncu@ricam.oeaw.ac.at}
\subtitle{\textit{\footnotesize To the memory of an inspirational mathematician, Marko Petkov\v{s}ek.}}



\begin{abstract}
  The Factorial Basis method, initially designed for quasi-triangular, shift-compatible factorial bases, provides solutions to linear recurrence equations in the form of definite-sums. This paper extends the Factorial Basis method to its $q$-analog, enabling its application in $q$-calculus. We demonstrate the adaptation of the method to $q$-sequences and its utility in the realm of $q$-combinatorics. The extended technique is employed to automatically prove established identities and unveil novel ones, particularly some associated with the Rogers-Ramanujan identities.
\end{abstract}

\maketitle

\section{Introduction}\label{sec:introduction}

Many combinatorics questions trickle down to finding a particular solution of a recurrence relation. This is no different in $q$-combinatorics. One of the most common appearances of sequences defined by a recurrence relation appears in the partition identities. A \textit{partition} $\pi$ is a non-increasing finite sequence $\pi = (\lambda_1,\lambda_2,\dots)$ of positive integers, called \textit{parts}. The \textit{size} of a partition is the sum of all its parts. Conventionally, the empty sequence is the unique partition of 0. A \textit{partition identity} is a statement that equates the number of partitions of size $n$ from two different sets of partitions. Maybe the best known and most recited partition identities are the Rogers--Ramanujan identities \cite{andrews1998theory}:\begin{theorem}[Rogers-Ramanujan identities, combinatorial version]\label{thm:RR_comb} Let $i=1$ or $i=2$. For every positive integer $n$, the number of partitions of size $n$ such that differences between consecutive parts are $\geq 2$ and the smallest part of the partition is $\geq i$ equals the number of partitions of size $n$ into parts congruent to $\pm i$ modulo 5.
\end{theorem}

There is an equivalent statement of the Rogers--Ramanujan identities in the $q$-series world:

\begin{theorem}[Rogers-Ramanujan identities, $q$-series version]\label{thm:RR_q}  Let $i=1$ or $i=2$. For $|q|<1$, we have \begin{equation}\label{eq:RR_GF}\sum_{n\geq 0}\frac{q^{n^2+(i-1)n}}{(q;q)_n} = \frac{1}{(q^i;q^5)_\infty(q^{5-i};q^5)_\infty},\end{equation}  
where 
$(a;q)_k := \prod_{j=0}^{k-1} (1-a q^j)$ for $k\in \mathbb{Z}_{\geq 0}\cup\{\infty\}$,
is the classical \textit{$q$-Pochhammer symbol}  \cite{andrews1998theory, gasper2011basic}.
\end{theorem}

The sides of the identity \eqref{eq:RR_GF} can be interpreted as the generating functions for the sets of partitions mentioned in Theorem~\ref{thm:RR_comb}. The summand $q^{n^2+(i-1)n}/(q;q)_n$, of the left-hand side of \eqref{eq:RR_GF}, is the generating function for the number of partitions with exactly $n$ parts, where the difference between parts is $\geq 2$ and the smallest part $\geq i$. Similarly, the right-hand side of \eqref{eq:RR_GF} is the generating function for the number of partitions into $\pm i$ modulo 5 parts. 

It is often important to have both the combinatorial and the $q$-series version of a partition identity. These different representations of the same problem increase the number of available proof techniques one can try. In fact, the Rogers--Ramanujan identities were first proven through $q$-series techniques, and, still to this day, there is no direct bijective proof.

In general, once a partition identity is identified through combinatorics the generating function equivalent is not apparent. It is usually easy to find generating functions for the number of partitions that only involve congruence conditions; those are infinite products. However, there is no standard way of constructing a generating function for the number of partitions with other conditions. 

While aiming to find such a generating function representation, it is common practice to pick a statistic of the partitions in question, bound that statistic, and form a recurrence relation over that. For example, the largest part of a partition.  Let $i=1$ or $2$, and consider $RR_i(N)$ to be the generating function of the partitions where differences between consecutive parts are $\geq 2$, the smallest part $\geq i$ and the largest part is $\leq N$. Hence, the limit of $RR_i(N)$ as $N \rightarrow \infty$ is the left-hand side sum of \eqref{eq:RR_GF}. Then one can easily check that $RR_i(N)$ satisfies the recurrence 
\begin{equation}\label{eq:RR_rec}
  RR_i(N) = RR_i(N-1) +  q^N RR_i(N-2),
\end{equation} with initial conditions $RR_i(0)=1$ and $RR_i(1)= 1+ q \delta_{i,1}$ \cite[Ch 3]{andrews1998theory}, where the exponent of $q$ counts the size of the partitions, and $\delta_{i,j}$ is the Kronecker delta function.

Our goal here is to construct an explicit formula for $RR_i(N)$ as a (nested) definite sum. To that end, we will consider a basis $B_{k}(N):=B_{k}(N;q)$ of $\mathbb{Z}[q]$ first and then find the recurrence $c^{(i)}_k:=c^{(i)}_k(q)$ satisfies, which is uniquely defined by the relation \begin{equation}\label{eq:RR_def_sum}RR_i(N) = \sum_{k=0}^N c^{(i)}_{k} B_{k}(N).\end{equation} Note that this method can be applied in a nested fashion. Once the recurrence and the initial values for $c^{(i)}_k$ are identified, one can pick a new basis and try to construct a definite sum representation of $c^{(i)}_k$ with respect to that basis, etc. We should also note that if one can find a sum representation of $RR_i(N)$, then by the limit $N\rightarrow \infty$, they can find a $q$-series representation of the generating function for the number of partitions that satisfy the difference conditions of Theorem~\ref{thm:RR_comb}.

We will be focusing on purely hypergeometric functions as summands. This ensures that we are always working with $q$-holonomic functions (i.e. functions that satisfy finite order linear recurrence relations with polynomial coefficients \cite{kauers2023Book}). Finding the recurrence satisfied by a definite sum with purely hypergeometric summands can be done using the \textit{Zeilberger's algorithm} (a.k.a. \textit{creative telescoping algorithm}) \cite{petkovvsek1996bak}. In that light, our approach is a way to attempt the \textit{inverse Zeilberger's problem} for $q$-series applications by user-guided experimentation. The technique to be developed in the later sections follows the footsteps of \cite{JimenezPastor2023}, in which the first author and Petkov\v{s}ek laid the foundational work for this method at the $q=1$ level. They called their approach \textit{factorial basis method}, and we will keep the name here. 


To give an example of how this method can and will be used, we look back at \eqref{eq:RR_def_sum}. One can choose the $q$-binomial basis ``$N$ choose $k$" for $B_{k}(N;q)$, where 
\begin{equation}\label{eq:q_binom}
  \qbinom{N}{k}=\left\{ \begin{array}{cl}
    \frac{(q;q)_N}{(q;q)_k(q;q)_{N-k}}, & \text{if}\ N\geq k\geq0,\\
    0, & \text{otherwise}.
\end{array} \right.
\end{equation}
Then, write the definite sum for $RR_i(N)$ 
\begin{equation}\label{eq:Shunsuke_Fin_sum}
  RR_i(N) = \sum_{k=0}^N c^{(i)}_k \qbinom{N}{k}
\end{equation} 
and factorial basis method finds and proves the recurrence and the initial values that define $c^{(i)}_k$ in \eqref{eq:RR_rec}. To be explicit, for this basis, we would get that $c^{(i)}_{k}$ satisfies the recurrence 
\begin{equation}\label{eq:Shunsuke_RR_rec} 
  c^{(i)}_k = -q^{k} c^{(i)}_{k-1} + q^{k-1} c^{(i)}_{k-2},
\end{equation} 
where $c^{(i)}_0 = 1$, $c^{(1)}_1= q$, and $c^{(2)}_1= 0$ for $i=1$ or 2. This is an intriguing new relation attached to a well-known and fundamental partition theorem. 
Moreover, by $N\rightarrow\infty$ and using \eqref{eq:RR_GF} we see the following theorem:
\begin{theorem}\label{thm:Shunsuke_RR} 
  For $i=1$ and 2, and $c^{(i)}_k$ defined as above, we have
  \begin{equation}\label{eq:Shunsuke_RR_thm}
    \sum_{k\geq 0} \frac{c^{(i)}_k}{(q;q)_k} = \frac{1}{(q^i;q^5)_\infty(q^{5-i};q^5)_\infty}.
  \end{equation}
\end{theorem}
Theorem~\ref{thm:Shunsuke_RR} was recently observed by Tsuchioka~\cite{tsuchioka2022fibonacci} for the first time through studying Kashiwara crystals. We independently discover this result by only using the original recurrence the bounded generating functions satisfy and the polynomial basis of our choice. Neither the statistic we bounded, nor the basis are unique choices. By picking another finitization of the generating functions and/or picking a different basis, we can discover similar results. 

Moreover, it is of interest to find a closed formula for $c^{(i)}_k$ to see~\eqref{eq:Shunsuke_RR_thm} as a $q$-series identity.  The recurrence~\eqref{eq:Shunsuke_RR_rec} is a $q$-analog of the well-known recurrence for Fibonacci numbers. The $q$-Fibonacci numbers were studied in the past \cite{carlitz1975fibonacci, andrews2004fibonacci, cigler2003q, cigler2004q}, some of which are in the context of Rogers--Ramanujan identities. We note a theorem by Cigler here~\cite[(4.2)]{cigler2003q}:
\begin{theorem}
\begin{equation}\label{eq:Cigler_formula} 
F_n(x,s,q) = \sum_{k=0}^{n-1}q^{k(k-1)}s^k x^{n-2k-1}\qbinom{n-k-1}{k}
\end{equation} satisfies $F_n(x,s,q)= x F_{n-1}(x,s,q) + s q^{n-3}F_{n-2}(x,s,q),$ with initial conditions $F_0(x,s,q)=0$ and $F_1(x,s,q)=1$. 
\end{theorem}
Although~\eqref{eq:Cigler_formula} is general, it is not possible to fit the recurrence in~\eqref{eq:Shunsuke_RR_rec} or some version of it after a substitution through this formula directly.  Instead, one can experiment with the factorial basis method on the sequence $c^{(i)}_k$ and discover a formula that way. This will lead to the following theorem:

\begin{theorem}\label{thm:Main_Shunsuke_thm} For  $k\geq 0$,
\begin{align}
\label{eq:2k_1}c^{(1)}_{2k}&= q^{k^2+k}\left(1- \qbinom{k}{1} - \sum_{l=1}^k\sum_{m=0}^{l-2} q^{l^2-lm+m^2+m}\qbinom{l-m-2}{m}\qbinom{k}{l+1}\right),\\
\label{eq:2kp1_1}c^{(1)}_{2k+1}&= q^{(k+1)^2}\left(1+ \sum_{l=1}^k\sum_{m=0}^{l-2} q^{l^2-lm+m^2-l+m}\qbinom{l-m-2}{m}\qbinom{k}{l}\right),\\
\label{eq:2k_2}c^{(2)}_{2k}&= q^{k^2+k}\left(1+ \sum_{l=0}^k\sum_{m=0}^{l-1} q^{l^2-lm+m^2+m}\qbinom{l-m-1}{m}\qbinom{k}{l+1}\right),\\
\label{eq:2kp1_2} c^{(2)}_{2k+1}&= - q^{(k+1)^2}\sum_{l=1}^k\sum_{m=0}^{l-1} q^{l^2-lm+m^2-l+m}\qbinom{l-m-1}{m}\qbinom{k}{l}.
\end{align}
\end{theorem}

We also note that it is possible to find formulas of these polynomial sequences \'{a} la Bressoud~\cite{bressoud1981some}.

\begin{theorem}\label{thm:Main_Shunsuke_thm2} For $\nu=0$ or 1, and  $k\geq 0$, we have
\begin{equation}\label{eq:C_explicit} c^{(i)}_{2k+\nu} = q^{k^2+k +\nu(k+1)}  \sum_{l=-\infty}^\infty \sum_{j=0}^4 q^{15l^2 }\alpha^{(i)}_{\nu,j} \qbinom{2k}{k+5l + j},\end{equation}
where 
\begin{align*}
\alpha^{(1)}_{0,0} &= q^{3l},\ \alpha^{(1)}_{0,1} = -q^{3l},\ \alpha^{(1)}_{0,2} = q^{13l+3},\ \alpha^{(1)}_{0,3} = 0,\ \alpha^{(1)}_{0,4} = -q^{23l+9},\\
\alpha^{(1)}_{1,0} &= q^{2l},\ \alpha^{(1)}_{1,1} = -q^{8l+1},\ \alpha^{(1)}_{1,2} =\alpha^{(1)}_{1,3} = \alpha^{(1)}_{1,4} = 0,\\
\alpha^{(2)}_{0,0} &= q^l,\ \alpha^{(2)}_{0,1} = -q^{9l+1},\ \alpha^{(2)}_{0,2} = q^{9l+1},\ \alpha^{(2)}_{0,3} = -q^{19l+6},\ \alpha^{(2)}_{0,4} = 0,\\
\alpha^{(2)}_{1,0} &= 0,\ \alpha^{(2)}_{1,1} = -q^{4l},\ \alpha^{(2)}_{1,2} = q^{14l+3},\ \alpha^{(2)}_{1,3} = 0,\ \alpha^{(2)}_{1,4} = 0.
\end{align*}
\end{theorem}

As mentioned in context, Theorems~\ref{thm:Main_Shunsuke_thm} and \ref{thm:Main_Shunsuke_thm2} provide new representations for the bounded generating functions $RR_i(N)$. 
Moreover, plugging the formulas \eqref{eq:2k_1}-\eqref{eq:2kp1_2} and \eqref{eq:Shunsuke_Fin_sum} for $c^{(i)}_k $ in \eqref{eq:Shunsuke_RR_thm}, we get new sum-product representations of the Rogers--Ramanujan identities. These can lead to new $q$-theoretic proofs of the identities~\eqref{eq:RR_GF}, as well as new combinatorial insights into Theorem~\ref{thm:RR_comb}.

The outline of the paper is as follows. In Section~\ref{sec:compatibility} we introduce the concept of a compatible operator as described in~\cite{JimenezPastor2023} and its associated recurrence operator. In Section~\ref{sec:qbasis} we extend the theory in~\cite{JimenezPastor2023} for factorial basis to fit the setting with $q$-series, introducing the concept of $\beta(n)$-factorial basis.
Section~\ref{sec:product} extends the product bases from~\cite{JimenezPastor2023} to the new $\beta(n)$-factorial bases. We then proceed to Section~\ref{sec:example}, proving Theorem~\ref{thm:Main_Shunsuke_thm} by using the adapted factorial basis method for $q$-series, together with other examples. In Section~\ref{sec:conclusions} we summarize our results and state possible future applications for this technique.

\section{Compatible operators}\label{sec:compatibility}

The fundamentals of the factorial basis method outlined in ~\cite{JimenezPastor2023} are needed here. We quickly recite the necessary concepts and adapt them to $q$-calculus as needed.

\begin{definition}\label{def:compatibility}
    Let $\KK$ be a field of characteristic zero and $\cB = \{B_k(n)\}_{k\in \NN}$ be a $\KK$-basis of $\KK^\NN$. Let 
    $L: \KK^\NN \rightarrow \KK^\NN$ be a $\KK$-linear operator, i.e., $L(a(n)+b(n)) = L(a(n)) + L(b(n))$ and 
    $L(\alpha a(n)) = \alpha L(a(n))$ for all $a(n), b(n) \in \KK^\NN$ and $\alpha\in \KK$. 

    We say that $L$ is $(A,B)$-compatible with $\cB$ in $t$ sections if there are sequences $\alpha_{r,i}(n) \in \KK^\NN$
    for $r=0,\ldots,t-1$ and $i=-A,\ldots,B$ such that, for all $k=mt+r$ it holds:
    \begin{equation}\label{equ:compatible}L P_k(n) = \sum_{i=-A}^B \alpha_{r,i}(m)P_{k+i}(n).\end{equation}
\end{definition}

In~\cite{JimenezPastor2023}, this definition was mainly used with $\KK=\QQ$, and several bases of $\QQ^\NN$ and their compatibilities were studied.
In this article, we will take $q$-series, i.e., $\KK = \QQ(q)$ where $q$ is a transcendental element over $\QQ$. Thus,
we will use the following notation:

\begin{itemize}
    \item $\cB = \{B_k(n)\}_{k\in \NN}$ is a $\QQ(q)$-basis of $\QQ^\NN$.
    \item Given $y(n) \in \QQ(q)^\NN$ we denote by $\sigma_\cB(y)$ the sequence $c(k) \in \QQ(q)^\NN$ such that $y(n) = \sum_k c_k B_k(n)$.
    \item Given $L$ a $(A,B)$-compatible operator with $\cB$ in 1 section, we denote by $\cR_\cB(L)$ (or simply $\cR(L)$) the linear recurrence operator defined by
    $\cR(L) = \sum_{i=-B}^A \alpha_{0,-i}(k) S^i$, where $S$ is the forward shift operator $S(a(k)) = a(k+1)$.
    \item Let $L$ be a $(A,B)$-compatible operator with $\cB$ in $t$ sections, we denote by $\cR_\cB(L)$ (or simply $\cR(L)$) to the $t\times t$ matrix of recurrence 
    operators whose element at position $(i,j)$ can be written as
    \[L_{j,r} = \sum_{\begin{array}{c}-A\leq i \leq B \\ i+r \equiv j~(mod~t)\end{array}} \alpha_{r,i}\left(m+\frac{j-r-i}{t}\right)S^{\frac{j-r-i}{t}}.\]
    When $t=1$, the $1\times 1$ matrix coincides with the recurrence operator defined above.
\end{itemize}

The following results, which are a rephrasing of results in~\cite{JimenezPastor2023}, relate the compatibility concept to 
the problem of finding definite-sum solutions to recurrence equations.

\begin{theorem}[{cf.~\cite[Theorem~2]{JimenezPastor2023}}]
    Let $L$ be $(A,B)$-compatible with $\cB$ in 1 section. Then there is a 1-1 correspondence between solutions of the equation $Ly(n) = 0$ and 
    $\cR(L) c(k) = 0$. More precisely, $L y(n) = 0$ if and only if $\cR(L) \sigma_\cB(y)(k) = 0$.
\end{theorem}

\begin{theorem}[{cf.~\cite[Proposition~39]{JimenezPastor2023}}]
  Let $L$ be $(A,B)$-compatible with $\cB$ in $t$ sections. Let $s_j$ denote the operation of taking the $j$-section module $t$ of a sequence. Then
  \[s_j \sigma_\cB(L y) = \sum_{r=0}^{t-1} L_{j,r}s_r\sigma_\cB (y).\]
\end{theorem}

In particular, if $Ly = 0$, then $0 = \sum_{r=0}^{t-1} L_{j,r}c(kt+r)$ for all $j = 0, \ldots, t-1$, relating the recurrence on the original sequence $y(n)$ with 
system of $t$ equations in $t$ sections of the induces sequence $c(k)$.

These two results from~\cite{JimenezPastor2023} are the key tool used to compute definite-sum closed form solution to recurrence equations (as was shown in that article).
In this paper we will use these ideas in the context of $q$-series, studying some bases of $q$-series, in particular the $q$-binomial bases.

From a computational perspective, the main appeal of this technique comes from the fact that $\cR_\cB$ is a ring homomorphism (see~\cite[Thm.~24 and Prop.~42]{JimenezPastor2023}). 
Namely, if $L_1$ and $L_2$ are two compatible operators in $t$-sections, then
\begin{itemize}
  \item $L_1+L_2$ is also compatible in $t$-sections with associated operator $\cR(L_1+L_2) = \cR(L_1)+\cR(L_2)$.
  \item $L_1L_2$ is also compatible in $t$-sections with associated operator $\cR(L_1L_2) = \cR(L_1)\cR(L_2)$.
\end{itemize}
Hence, compatibilities can be computed for recurrence operators $\QQ(q, q^n)\langle E_n\rangle$ by simply studying the compatibilities of the operators $q^n$ (multiplication 
by the sequence $q^n$) and the shift operator $E_n: n\mapsto n+1$.

\section{\texorpdfstring{$\beta(n)$-factorial basis and $q$-binomial bases}{b(n)-factorial basis and q-binomial bases}}\label{sec:qbasis}

In~\cite{JimenezPastor2023}, factorial bases were studied. Those were bases $B_k(n)$ such that they satisfy a linear first order recurrence $B_{k+1}(n) = (a(k)n + b(k))B_k(n)$. In the 
context of $q$-series this is usually not enough: usually, the recurrences involve the sequence $q^n$ instead of simply $n$.

\begin{example}[$q$-binomial basis]
  Let us consider $B_k(n) = \qbinom{n}{k}$ as defined in~\eqref{eq:q_binom}. Hence, by construction, we have:
  \begin{align*}
    B_{k+1}(n) & = \frac{(q;q)_n}{(q;q)_{k+1}(q;q)_{n-k-1}} = \frac{1-q^{n-k}}{1-q^{k+1}}\frac{(q;q)_n}{(q;q)_{k}(q;q)_{n-k}} \\
               & = \frac{1-q^{n-k}}{1-q^{k+1}} B_k(n) = \left(-\frac{1}{q^k(1-q^{k+1})}q^n + \frac{1}{1+q^{k+1}}\right)B_k(n).
  \end{align*}
\end{example}

However, the similar structure between these recurrences and those studied in the factorial bases, allows us to adapt the theory from one to the other.

\begin{definition}\label{def:factorial}
  Let $\beta(n)$ be a $q$-sequence. We say that the set $\{B_k(n)\}$ is $\beta(n)$-factorial w.r.t. $k$ if there are sequences $a(k), b(k)$ such that 
  \[B_{k+1}(n) = (a(k)\beta(n) + b(k))B_k(n).\]
\end{definition}

When $\beta(n) = n$, this concept coincides with the definition of factorial basis in~\cite{JimenezPastor2023}. Moreover, we can see a $\beta(n)$-factorial basis 
as a polynomial ring of sequences embedded into the sequences $\QQ(q)^\NN$. More precisely, if we consider the sequence of polynomials $p_k(Y) \in \QQ(q)[Y]$ defined
by $p_{k+1}(Y) = (a(k)Y+b(k))p_k(Y)$, then the map $\mathcal{E}: \QQ(q)[Y] \rightarrow \QQ(q)^\NN$ defined by $\mathcal{E}(Y) = \beta(n)$ satisfies that
\begin{itemize}
  \item If $\beta(n) \in \QQ(q)^\NN$ is transcendental over $\QQ(q)$, $\mathcal{E}$ is injective.
  \item $\mathcal{E}(p_k(Y)) = B_k(n)$.
\end{itemize}

Hence, the set $\{B_k(n)\}$ is linearly independent set. Moreover, these sets enjoy a very simple compatibility condition by multiplication by $\beta(n)$:

\begin{lemma}
  Let $\{B_k(n)\}_{k\in\NN}$ be a $\beta(n)$-factorial basis. Then the linear operator $X: \QQ(q)^\NN \rightarrow \QQ(q)^\NN$ defined by $X(a(n)) = a(n)\beta(n)$ is $(0,1)$-compatible in 1 section.
\end{lemma}
\begin{proof}
  By definition, we know that there are sequences $a(k)$ and $b(k)$ such that $B_{k+1}(n) = (a(k)\beta(n) + b(k))B_{k}(n)$, so, rearranging terms we have
  \[X (B_k(n)) = \beta(n)B_k(n) = \frac{1}{a(k)} B_{k+1}(n) - \frac{b(k)}{a(k)} B_k(n).\]
\end{proof}

\begin{example}[$q$-Power Basis]\label{exm:power}
  Let $e\in \NN$ and consider the $q$-power basis $\cP^{(e)} = \{q^{ekn}\}_k$. This set is $(q^{en})$-factorial, where $a(k) = 1$ and $b(k) = 0$. Hence, $\cP$ is $(0,1)$-compatible with the multiplication by $q^{en}$ leading to the recurrence operator $\cR_{\cP^{(e)}}(q^{en}) = S_k^{-1}$.
\end{example}

As we have remarked before, in a $\beta(n)$-factorial basis the sequence $B_k(n)$ is always a polynomial of degree $k$ w.r.t. $\beta(n)$ (by taking $p_k(Y) = \mathcal{E}^{-1}(B_k(n))$). This polynomial
has $k$ roots in $\QQ(q)$. By construction, the roots of $p_k(Y)$ are the roots of $p_{k-1}(Y)$ together with a new root built from $a(k)$ and $b(k)$. More precisely, for $k \geq 1$, the new root of 
$p_k(Y)$ is $\rho(k-1) = -b(k-1)/a(k-1)$.

Hence, we can define a $\beta(n)$-factorial basis by defining its root sequence and providing a new sequence $c(k)$ for the leading coefficient of the polynomial 
w.r.t. $\beta(n)$. 

\begin{example}[$q$-Falling Basis]\label{exm:falling}
  Let us consider the $q^n$-factorial basis defined by the root sequence $\rho(k) = q^{k}$ and the leading coefficient sequence $c(k) = 1$. 
  This induced basis $\cF = \{f_k(n)\}_k$ (which we will call $q$-falling basis) can be written as
  \[f_k(n) = \prod_{i=0}^{k-1} (q^n - q^{i}).\]
  This basis satisfies a straightforward compatibility with the multiplication by $q^n$. Following the notation on Definition~\ref{def:factorial}, 
  the associated sequences $a(k)$ and $b(k)$ to the basis $\cF$ are $a(k) = 1$ and $b(k) = q^{k}$. Hence, $q^n f_k(n) = f_{k+1}(n) + q^{k}f_k(n)$,
  leading to
  \[\cR_\cF(q^n) = S_k^{-1} + q^k.\]

\end{example}
This approach using the root sequence allows proving compatibility by a shift operator using only the associated root sequence $\rho(k)$.

\begin{lemma}[{cf.~\cite[Proposition~11]{JimenezPastor2023}}]\label{lem:char}
  Let $\cB$ be a $\beta(n)$-factorial basis and let $p_k(Y)\in \QQ(q)[Y]$ be such that $B_k(n) = p_k(\beta(n))$. Then $L$ is $(A,B)$-compatible with $\cB$ if
  and only if
  \begin{enumerate}
    \item For all $k \in \NN$, $L B_k(n) \in \QQ(q)[\beta(n)]$. Let us denote $q_k(Y)$ the polynomial such that $L B_k(n) = q_k(\beta(n))$.
    \item For all $k \in \NN$, $deg_{Y} q_k(Y) \leq k+B$.
    \item For all $k \in \NN$, $p_{k-A}(Y)$ divides $q_k(n)$.
  \end{enumerate}
\end{lemma}
\begin{proof}
  First, let us assume that $L$ is a $(A,B)$-compatible operator with $\cB$. Then, the application of $L$ to $B_k(n)$ is a linear combination of $B_{k-A}(n),\ldots,B_{k+B}(n)$. In terms
  of polynomials in $\beta(n)$, this is the same as being a linear combination of $p_{k-A}(Y),\ldots,p_{k+B}(Y)$. Hence,
  \begin{itemize}
    \item $L B_k(n)$ is a polynomial in $\beta(n)$.
    \item The degree of $q_k(Y)$ is at most $\deg_Y(p_{k+B}(Y)) = k+B$.
    \item Since $\cB$ is $\beta(n)$-factorial, then $p_k(Y)$ divides $p_K(Y)$ for all $k \leq K$. Hence, $p_{k-A}(Y)$ divides $q_k(Y)$.
  \end{itemize}

  Now, assume the three points in the statement hold for $L$ and write $q_k(Y) = \sum_{i=0}^{k+B}c_i(k) p_i(Y)$ (since the polynomials $p_k(Y)$ form a basis of polynomials). Since
  $p_{k-A}(Y)$ divides $q_k(Y)$ by hypothesis and also divides $\sum_{i=k-A}^{k+B} c_i(k)p_i(Y)$, then it also divides the polynomial $\sum_{i=0}^{k-A-1}c_i(k) p_i(Y)$. Since the degree 
  of this polynomial is strictly smaller than the degree of $p_{k-A}(Y)$, we conclude it is the zero polynomial, i.e., $c_i(k) = 0$ for $i=0,\ldots, k-A-1$.
  Hence, 
  \[q_k(Y) = \sum_{i=k-A}^{k+B}c_i(k) p_i(Y) = \sum_{i=-A}^Bc_{k+i}(k) p_{k+i}(Y),\]
  showing that $L$ is $(A,B)$-compatible with $\cB$.
\end{proof}

\begin{lemma}\label{lem:char_E_root}
  Let $\cB = \{B_k(n)\}_k$ be a $\beta(n)$-factorial basis with root sequence $\rho(k)$. Let $E: \QQ(q)^\NN \rightarrow \QQ(q)^\NN$ be a ring homomorphism such that
  $E\cdot \beta(n) = \gamma \beta(n) + \nu$ for $\gamma,\nu \in \QQ(q)$. Then $E$ is $(A,0)$-compatible with $\cB$ if and only if for all $n\in \NN$ the following multiset 
  inclusion holds:
  \[\left[\gamma\rho(0) + \nu, \ldots, \gamma\rho(n) + \nu\right] \subseteq [\rho(0),\ldots,\rho(n+A)].\]
\end{lemma}
\begin{proof}
  Following the notation on Lemma~\ref{lem:char}, the operator $E$ acts on $\QQ(q)[Y]$ by the rule $E(Y) = \gamma Y + \nu$. Hence, $q_k(Y) = p_k(\gamma Y+ \nu)$.

  If $E$ is $(A,B)$-compatible with $\cB$, then $p_{k-A}(Y)$ divides $q_k(Y)$, so all the roots of $p_{k-A}(Y)$ must appear in $q_k(Y)$, implying the following inclusion of multisets:
  \[[\rho(0),\ldots,\rho(k-A-1)] \subseteq \left[\frac{\rho(0)-\nu}{\gamma},\ldots,\frac{\rho(k-1)-\nu}{\gamma}\right],\]
  which is equivalent to the inclusion of multisets
  \[[\gamma\rho(0) + \nu, \ldots, \gamma\rho(k-A-1) + \nu] \subseteq [\rho(0),\ldots,\rho(k-1)].\]
  Finally, since this holds for all $k \geq A$, then we have for all $k \in \NN$:
  \[[\gamma\rho(0) + \nu, \ldots, \gamma\rho(k) + \nu] \subseteq [\rho(0),\ldots,\rho(k+A)].\vspace{-4mm}\]
\end{proof}

\begin{example}[Continuation of Example~\ref{exm:power}]\label{exm:power:2}
  Recall the $q$-power basis defined in Example~\ref{exm:power} was $(q^{en})$-factorial with $a(k) = 1$ and $b(k) = 0$. Then, the root sequence is $\rho(k) = 0$ and, using Lemma~\ref{lem:char_E_root}, $\cP^{(e)}$ is clearly $(0,0)$-compatible with $E$:
  \[E  q^{ekn} = q^{ek(n+1)} = q^{ek} q^{ekn}.\]
  Obtaining the following associated recurrence $\cR_{\cP^{(e)}}(E) = q^k$.
\end{example}

\begin{example}[Continuation of Example~\ref{exm:falling}]
  For the $q$-Falling basis $\cF$, we know its root sequence is $\rho(k) = q^{k}$. Then, using Lemma~\ref{lem:char_E_root}, $\cF$ is $(1,0)$-compatible with the shift operator $E$. Namely,
  \begin{align*}
    E f_k(n) & = f_k(n+1) = \prod_{i=0}^{k-1} (q^{n+1} - q^i) = q^{k-1}(q^{n+1}-1) f_{k-1}(n) \\
    & = q^{k-1} (f_{k-1}(n)q^nq - f_{k-1}(n)) \\
    & = q^{k-1} (q^{k}-1) f_{k-1}(n) + q^kf_k(n).
  \end{align*}
  Then the recurrence operator associated with $E$ w.r.t. $\cF$ is $\cR_\cF(E) = q^{k} + q^k(q^{k+1} -1)S_k$.
\end{example}

We will study now the compatibilities associated with the $q$-binomial basis. The following result shows a generic result for a family of $q$-binomial bases.

\begin{lemma}\label{lem:qbinom_factorial}
  Let us define the sequence $\cC(a,c;t;e) = \left\{\qbinom[q^e]{an+c}{k+t}\right\}_k$ for $a, c, t, e\in \ZZ$. Then $\cC(a,c;t;e)$ is $q^{aen}$-factorial.
\end{lemma}
\begin{proof}
  Let $C_k(n) = \qbinom[q^e]{an+c}{k+t}$. By definition of the $q$-binomial, we have that 
  \begin{align*}
    C_{k+1}(n) & = \qbinom[q^e]{an+c}{k+t+1} = C_{k}(n)\frac{1-q^{e(an+c-k-t)}}{1-q^{e(k+t+1)}}\\
               & = C_{k}(n) \left(q^{aen}\frac{-q^{e(c-k-t)}}{1-q^{e(k+t+1)}} + \frac{1}{1-q^{e(k+t+1)}}\right),
  \end{align*}
  which shows that $\cC(a,c;m;e)$ is $q^{aen}$-factorial with coefficients 
  \[a(k) = \frac{-q^{e(c-k-t)}}{1-q^{e(k+t+1)}},\quad b(k) = \frac{1}{1-q^{e(k+t+1)}}.\vspace{-3mm}\]
\end{proof}

From this proof, we can see that the root sequence for the basis $\cC(a,c;t;e)$ is $\rho(k) = -b(k)/a(k) = q^{e(k+t-c)}$. Moreover, we can see that
$E: n\mapsto n+1$ acts nicely over $q^{aen}$, since $E \cdot q^{aen} = q^{ae}q^{aen}$. We can now use Lemma~\ref{lem:char_E_root} to show that $E$ is always
$(a,0)$-compatible with the generic $q$-binomial basis.

\begin{lemma}[Compatibility with shift]\label{lem:qbinom_comp_E}
  Let us denote with $C_k(n) = \qbinom[q^e]{an+c}{k+t}$. Then $E: n \mapsto (n+1)$ is $(a,0)$-compatible with $\cC(a,c;m;e)$.
\end{lemma}
\begin{proof}
  Using Lemma~\ref{lem:char_E_root}, we only need to show that the transformed root sequence up to $k$ elements is contained in the root sequence
  up to $k+a$ element. Recall that the root sequence is $\rho(k) = q^{e(k+t-c)}$, so its transformed version is $q^{e(a + k + t - c)}$. Then, we need to 
  check that for all $k \in \NN$
  \[[q^{e(a+l+t-c)}\ :\ l=0,\ldots,k] \subseteq [q^{e(l+t-c)}\ :\ l=0,\ldots,k+a],\]
  which can be easily check by the reader.
\end{proof}

\begin{example}[{cf.~\cite[Theorem~3.3]{andrews1998theory}}]
  Let us use this technique to prove the following classical result
  \begin{equation}\label{eq:Andrews_binom}(z;q)_{n} = \sum_{k=0}^n (-1)^kz^kq^{(k^2-k)/2} \qbinom{n}{k}.\end{equation}

  First, we need the recurrence operator that annihilates the sequence $(z;q)_n$. This is easy by using the definition of the Pochhammer symbol, yielding
  \[(E - 1 + zq^n) \cdot (z;q)_n = 0.\]

  Secondly, we want to express the sequence $(z;q)_n$ in terms of a sum over the $q$-binomial $\qbinom{n}{k}$. Hence, we need to know if our recurrence operator
  $(E - 1 + zq^n)$ is compatible with this basis. Using Lemma~\ref{lem:qbinom_factorial} and~\ref{lem:qbinom_comp_E}, we know the $q$-binomial basis is $q^n$-factorial 
  (hence, compatible with multiplication by $q^n$) and also $(1,0)$-compatible with $E$. More precisely,
  \begin{align*}
    q^n \qbinom{n}{k} & = q^k\qbinom{n}{k} + q^k(q^{k+1}-1)\qbinom{n}{k+1},\\
    E  \qbinom{n}{k} & = \qbinom{n}{k-1} + q^k\qbinom{n}{k}.
  \end{align*}

  Then we can build the transformed operators:
  \begin{align*}
    \cR(q^n) = & q^{k-1}(q^k - 1) S_k^{-1} + q^k\\
    \cR(E) = & q^k + S_k
  \end{align*}

  Using the fact that $\cR$ acts as an operator homomorphism, we then have that
  \[\cR(E - 1 + zq^n) = \frac{zq^k}{q}(q^k-1)S_k^{-1} + ((z+q)q^k - 1) + S_k.\]

  Hence, if we write $(z;q)_n = \sum_k c(k)\qbinom{n}{k}$, then we know $c(k)$ is annihilated by $\cR(E - 1 + zq^n)$. We can now use
  classical holonomic techniques to show that $c(k) = (-1)^kz^kq^{(k^2-k)/2}$.
\end{example}

\section{Product bases}\label{sec:product}

In Section~\ref{sec:qbasis} we studied the compatibility of $\cP^{(1)}, \cF$ and $\cC(1,0;0;1)$ with the operators $q^n$ and $E$. This would allow us to compute
solutions with nested sums like
\[\sum_{i}\sum_{j}\sum_{k} c(k) q^{kj} f_{j}(i) \qbinom{n}{i},\]
for any recurrence operator involving the shift operator $E: n\mapsto (n+1)$ and the multiplication by the sequence $q^n$. In this type of nested sums we see the 
indices do not match between different bases and this restricts the possibilities of finding nice definite closed form solutions to our original recurrence equations.

Sometimes we would like to have a solution that has the form $\sum_k c(k)q^{kn}\qbinom{n}{k}$. This problem was already tackled in~\cite{JimenezPastor2023}
with the use of \emph{Product bases}. Let us consider two $\beta(n)$-factorial bases for the same sequence $\beta(n)$. These bases have associated the two
sequences $a_1(k), b_1(k)$ and $a_2(k), b_2(k)$. We can then combine the two bases by interlacing the sequences $a_1(k)$ with $a_2(k)$ and $b_1(k)$ with
$b_2(k)$.

Let us denote by $P_k(k)$ the elements of the first basis (defined using $a_1(k)$ and $b_1(k)$) and by $Q_k(n)$ the elements of the second basis (defined
using $a_2(k)$ and $b_2(k)$). Then the resulting $\beta(n)$-factorial basis using $a(k)$ (resp. $b(k)$) the interlacing sequence of $a_1(k)$ and $a_2(k)$ (resp.
of $b_1(k)$ and $b_2(k)$) is the sequence
\[1,\quad P_1(n),\quad P_1(n)Q_1(n),\quad P_2(n)Q_1(n),\quad P_2(n)Q_2(n),\quad \ldots.\]

It was shown in~\cite{JimenezPastor2023} that this idea can be generalized to more than 2 basis and even with a different type of interlacing (see \emph{Shuffled
basis} in~\cite{JimenezPastor2023}). But for the purpose of this paper, we will simply recall the results that allow to extend the compatibility from 
the $\beta(n)$-factorial bases to the corresponding Product Basis.

\begin{theorem}[{cf.~\cite[Theorem~34 and Theorem~55]{JimenezPastor2023}}]\label{thm:product}
  Let $\cB_1,\ldots,\cB_m$ be $\beta(n)$-factorial bases. Then
  \begin{enumerate}
    \item The product basis $\cB = \prod_{i=1}^m \cB_i$ is $\beta(n)$-factorial.
    \item Let $L$ be a ring homomorphism of $\QQ(q)^\NN$ and assume $L$ is $(A_i,B_i)$-compatible with each $\cB_i$ in $t_i$ sections.
    Then $L$ is $(mA,B)$-compatible in $mt$ sections with $\cB$, where $A = \max\{A_i\}$, $B = \max\{B_i\}$ and $t = \lcm\{t_i\}$.
  \end{enumerate}
\end{theorem}

\begin{example}[Product of $(q^{2n})$-factorial bases]\label{Example_Product}
  Let us consider the bases $\cP^{(2)}$ and $\cC(1,0;0;2)$. They are $(q^{2n})$-factorial bases and let us denote by $P_k(n)$ and $B_k(n)$ their respective elements, i.e., 
  \[P_k(n) = q^{2kn},\qquad Q_k(n) = \qbinom[q^2]{n}{k}.\]
  Let us consider the product bases of the two, namely
  \[B_{2k}(n) = P_k(n)Q_k(n),\qquad B_{2k+1}(n) = P_{k+1}(n)Q_k(n).\]

  Using Examples~\ref{exm:power} and~\ref{exm:power:2}, and Lemmas~\ref{lem:qbinom_factorial} and~\ref{lem:qbinom_comp_E}, we know both bases are 
  compatible with the multiplication by $q^{2n}$ and the shift operator $E$. Hence, by Theorem~\ref{thm:product}, we know that the product basis 
  is again $(q^{2n})$-factorial and is $(1,0)$-compatible with $E$. Let us detail the compatibility conditions with both $q^{2n}$ and $E$.

  The compatibility with $q^{2n}$ extends naturally to the product basis as follows:
  \begin{align*}
    q^{2n} B_{2k}(n)   & = q^{2n}P_k(n)Q_k(n) = P_{k+1}(n)Q_k(n) = B_{2k+1}\\
    q^{2n} B_{2k+1}(n) & = P_{k+1}(n) q^{2n}Q_k(n) \\
                       & = P_{k+1}(n) \left(q^{2k}Q_k(n) + q^{2k}(q^{2k+2}-1)Q_{k+1}(n)\right) \\
                       & = q^{2k} B_{2k+1}(n) + q^{2k}(q^{2k+2}-1)B_{2k+2}(n).
  \end{align*}

  The compatibility of $E$ is also extended by using the compatibility of $E$ with the basis $P_k(n)$ and $Q_k(n)$:
  \begin{align*}
    E B_{2k}(n)   & = P_k(n+1) Q_k(n+1) \\
                  & = (q^{2k}P_k(n)) (Q_{k-1}(n) + q^{2k}Q_{k}(n)) \\
                  & = q^{2k}P_k(n)Q_{k-1}(n) + q^{4k}P_k(n)Q_{k}(n) \\
                  & = q^{2k}B_{2k-1}(n) + q^{4k}B_{2k}(n).\\
    E B_{2k+1}(n) & = P_{k+1}(n+1) Q_k(n+1)\\
                  & = (q^{2k+2}P_{k+1}(n)) (Q_{k-1}(n) + q^{2k}Q_{k}(n))\\
                  & = (q^{2k+2}P_{k+1}(n)Q_{k-1}(n) + q^{4k+2}P_{k+1}Q_{k}(n))\\
                  & = (q^{2k+2}q^{2n}P_{k}(n)Q_{k-1}(n) + q^{4k+2}P_{k+1}Q_{k}(n))\\
                  & = q^{2k+2}q^{2n} B_{2k-1}(n) + q^{4k+2}B_{2k+1}(n)\\
                  & = q^{4k} B_{2k-1}(n) + q^{4k}(q^{2k}-1)B_{2k}(n) + q^{4k+2}B_{2k+1}(n)
  \end{align*}
  Confirming that $E$ is $(2,0)$-compatible with the product basis. Combining everything, we have the following recurrence matrices associated with these two
  compatibilities:
  \begin{align*}
    \cR(q^{2n}) & = \begin{pmatrix}0 & q^{2k-2}(q^{2k}-1) S_k^{-1}\\1& q^{2k}\end{pmatrix},\\
    \cR(E)      & = \begin{pmatrix}q^{4k} & q^{4k}(q^{2k}-1)\\ q^{2k+2}S_k & q^{4k+2} + q^{4k+4}S_k\end{pmatrix}.
  \end{align*}
\end{example}
  
\section{Examples}\label{sec:example}

In his comprehensive list of finitizations Sills~\cite{Sills_Fin} listed many Rogers--Ramanujan type identities. We can tackle these examples with the factorial method. We will give an example here and showcase the use cases of this method.

\begin{example}[Identity 3.38 of \cite{Sills_Fin}]\label{eg:Sills_example} Let $a_n$ be defined by
\[ \left(E^2 - (1+q) E- (q^{2n+4}-q)\right)a_n = 0, \] in operator notation, with the initial values $a_0=1$, and $a_1 = 1+q$. Factorial basis method allows us to experiment and search for a suitable basis. For example, we can try the $q$-Binomial basis $B_{k}(n):=\qbinom{n}{k}$ and expand $a_n = \sum_{k\geq 0 } a'_{k} B_{k}(n)$. Factorial basis method determines that $a'_k$ is annihilated by the operator \begin{align*}
S_k^4 &-(1+q) (1-q^{k+2})S_k^3 - (q^{2k+8} - q^{2k+4} + q^{k+3} + q^{k+2} - q)S_k^2\\ &+q^{2k+6}(1+q) (1-q^{k+2})S_k - q^{2k+5}(1+q^{k+1}) (1-q^{k+2}),
\end{align*} 
and has the initial conditions $a'_0 =1,$ $a'_1 =q,$ $a'_2 =q^4,$ and $a'_3=q^7$. Using this recurrence and \texttt{qFunctions}~\cite{ablinger2021qfunctions}, we can guess and prove that the $a'_k$ gets annihilated by the much simpler operator \begin{equation}\label{eq:double_shift}S_k^2 - q^{2n+4}.\end{equation} Hence, we can easily write a formula for $a'_k$ and get a closed formula for $a_n$ with respect to the basis $B_{k}(n)$. 

Furthermore, the double shift in the operator \eqref{eq:double_shift} suggests that a different basis where this double shift is inherent might be a better choice. To that end, we can look at sections of $B_{k}(n)$. The section of $B_{k}(n)$ corresponding to $\qbinom{n}{2k}$, which corresponds to writing $a_n = \sum_{k\geq 0 } a''_{k} B_{2k+1}(n)$, shows that \[\left(S_k - q^{4n+4}\right) a''_{k} =0\] and has the initial conditions $a''_{0}=1$. Therefore, we get \begin{equation}\label{eq:Sills_pn338}a_n = \sum_{k=0}^\infty q^{2k(k+1)} \qbinom{n+1}{2k+1}.\end{equation}  This is the sum representation in \cite{Sills_Fin}. Moreover, the other sections $B_{2k}(n)$, $B_{2k+1}(n)$, and $B_{2k}(n+1)$ fail to calculate corresponding summands and certify that these choices do not yield sum representations.
\end{example}

We note that the factorial basis method finds an Ore operator that would annihilate the corresponding summand for any basis and section, indiscriminately. Then solving the summands explicitly with respect to the bases using the initial conditions of the original sequence shows us if there exists such a summand. For example, one can check that if we were to change $a_1=1$ in the Example~\ref{eg:Sills_example}, then only the $B_{2k}(n)$ section will yield a solution.

As it was also indicated by the Rogers--Ramanujan example, the factorial basis method can be applied to well-known partition theory results. These can be led to interesting theorems similar to Theorem~\ref{thm:Shunsuke_RR}. To demonstrate that, we recite a generating function due to Andrews--Alladi--Gordon, found using the method of weighted words in their proof and refinement of Capparelli's identity \cite{AlladiAndrewsGordon_Capparelli}.

\begin{example}\label{ex:AAG} Let $G_{N}(a,b;q)$ be the generating function for partitions into parts $>1$ and $\leq 3N-2$ with minimal difference between its parts is $\geq 2$, where the difference between consecutive parts $\geq 4$ unless consecutive parts add up to a multiple of 3, where the exponent of $a$, $b$ and $q$ counts the number of 1 modulo 3 parts, the number of 2 modulo 3 parts, and the total size of the partition, respectively. Then,  we have \cite[Lemma 1]{AlladiAndrewsGordon_Capparelli}
\begin{equation}\label{eq:AAG_GN} G_{N}(a,b;q) = \sum_{j=0}^\infty q^{3\binom{N-2j}{2}}\qbinom[q^3]{N}{2j} (-a q^4;q^6)_j (-bq^2; q^6)_j.\end{equation} We know \cite[(4.1)]{AlladiAndrewsGordon_Capparelli} that $G_N:=G_{N}(a,b;q)$ satisfies the following recurrence
\begin{align*}
G_{N+1}&=(1+q^{3n})G_{N} + (b q^{3n-1} + a q^{3n+1} + a b q^{6n})G_{N-1}\\
&+abq^{6n-3}(1-q^{3n-3})G_{N-2},
\end{align*}
with initial conditions $G_0=G_1=1$, and $G_2 = 1+ bq^2 + q^3 + a q^4 + a b q^6$. Furthermore, it is also known \cite[(4.7, 4.8)]{AlladiAndrewsGordon_Capparelli} that \begin{equation}\label{eq:AAG_lim}\lim_{N\rightarrow \infty} G_N(a,b;q) = (-q^3;q^3)_\infty(-aq^4;q^4)_\infty (-bq^2;q^6)_\infty.\end{equation}

Picking the binomial basis with $q^3$ (i.e., the basis $\cC(1,0;0;3)$ from Lemma~\ref{lem:qbinom_factorial}) and splitting it in two sections, we can write $G_N(a,b;q)$ as 
\begin{equation}\label{eq:AAG_GN_FacBasis}
  G_N(a,b;q) = \sum_{j=0}^\infty d_j(a,b;q) \qbinom[q^3]{N}{2j}.
\end{equation} 
Factorial basis method shows that $d_j:=d_j(a,b;q)$ is annihilated by the operator 
\begin{align}\label{eq:rAAG_ec_ck}
    S_j^3 - q^{6j+14}&(b+aq^2+q^{6j+13}+abq^{6j+16}  )S_j^2 \\ \nonumber
    &+ ab q^{12j+24}(1-q^{6j+9})(1-q^{6j+12})S_j.
\end{align} 
This operator, once again, does not correspond to the minimal recurrence $d_j$ satisfies. 

Here it is easy to observe that there is an extra shift operator $S_j$ common in the terms of the annihilator~\eqref{eq:rAAG_ec_ck}. Removing this common term yields the minimal recurrence satisfied by $d_j$. In general, as in Example~\ref{eg:Sills_example}, the minimal recurrence can be found and shown to annihilate $d_j$ using the proven recurrence, through \texttt{qFunctions}~\cite{ablinger2021qfunctions}. 

Hence, we have that $d_j$ is defined by the recurrence \begin{align}\label{eg:AAG_rec_shunsuke} d_j = q^{6j-4}&(b+q^{6j-5}+aq^2+ a b q^{6j-2})d_{j-1}\\ \nonumber &+ ab q^{12j-12} (1 - q^{6 j - 9}) (1 - q^{6 j - 6})d_{j-2}
\end{align} with the initial conditions $d_0=1$ and $d_1=bq^2 + q^3 + a q^4 + a b q^6$. 
\end{example}

Example~\ref{ex:AAG}, yields a theorem in the spirit of Theorem~\ref{thm:Shunsuke_RR}.

\begin{theorem} The $d_j$ defined as in \eqref{eg:AAG_rec_shunsuke}, we have
\[\sum_{j\geq 0} \frac{d_j}{(q^3;q^3)_{2j}} =(-q^3;q^3)_\infty(-aq^4;q^4)_\infty (-bq^2;q^6)_\infty. \]
\end{theorem}

The proof of the above statement comes from taking the limit $N\rightarrow\infty$ of \eqref{eq:AAG_GN_FacBasis} and comparing this limit with \eqref{eq:AAG_lim}.

Finding a closed formula for $d_j$ would also have $q$-series implications when comparing \eqref{eq:AAG_GN} and \eqref{eq:AAG_GN_FacBasis}. The recurrence \eqref{eg:AAG_rec_shunsuke} is of order 2 and $d_j(a,b;q)$ does not have a d'Alembertian solution. One can try to apply factorial basis for $d_j$ and try to construct $d_j$ as a nested definite sum. Here, we will look at a particular substitution and see a corollary of $d_j$. Fix $b=-q$, then we get \begin{equation}\label{eq:d_j_Subd}d_j(a,-q;q) = a^jq^{3j^2+j}(q^3;q^6)_j,\end{equation} automatically by Sigma \cite{schneider2004summation}. Comparing \eqref{eq:AAG_GN} with $b\mapsto -q$ and \eqref{eq:AAG_GN_FacBasis} with \eqref{eq:d_j_Subd} yields
\begin{corollary}\label{cor:CO} For any non-negative $N$, we have
\begin{equation*}
\sum_{j=0}^\infty q^{3\binom{N-2j}{2}}\qbinom[q^3]{N}{2j} (-a q^4;q^6)_j (q^3; q^6)_j = \sum_{j=0}^\infty a^j q^{3j^2+j}\qbinom[q^3]{N}{2j}(q^3;q^6)_j.
\end{equation*}
\end{corollary}
Although both summands visually look similar (which is a by-product of our choice of the factorial basis section), a $q$-theoretic proof of such a result requires some work. To demonstrate that, we will outline its $q$-theoretic proof here: as a first step, one can extract the exponents of $a$ on both sides of this identity (on the left-hand side by \eqref{eq:Andrews_binom}) and compare these terms. Simplifying both sides yield \begin{equation}\label{eq:Cor_mid_step}\frac{1}{(q;q)_{N-2k}} = \sum_{j=0}^\infty \frac{q^{\frac{3}{2}(N-2j)(N-2j-1)}}{(q;q)_{N-2j}(q^2;q^2)_{j-k}},\end{equation} where $N\geq 2k\geq 0$. We recall that $1/(q;q)_k = 0$ when $k$ is a negative integer. Making the substitutions $i=j-k$ and $M=N-2k$, using \cite[I.10]{gasper2011basic}, $q\mapsto 1/q$, using \cite[I.3]{gasper2011basic} in succession while doing cancellations as needed, we get
\[q^{\frac{M^2-M}{2}} = \sum_{i=0}^\infty \frac{(q^{-M};q)_{2i}}{(q^2;q^2)_i}(-1)^i q^{-i^2 + 2iM}.\] This last identity can be seen as a consequence of \cite[II.7]{gasper2011basic}, the $q$-Chu--Vandermonde sum, as the right-hand side sum is equal to $\lim_{\rho\rightarrow\infty} {}_2\phi_1(q^{-M},q^{-M+1}, \rho ; q^2, \rho q^{2M-1})$, where ${}_2\phi_1(a,b,c;q,z)$ is the classical basic hypergeometric series as defined in \cite{gasper2011basic}. Following these manipulations backward yields a proof of Corollary~\ref{cor:CO}.

Now we present the proof of Theorem~\ref{thm:Main_Shunsuke_thm}, with commentary on the use of factorial basis and other relevant methods, as the next example.

\begin{example}[Proof of Theorem~\ref{thm:Main_Shunsuke_thm}]
Looking at the first few terms of $c^{(i)}_k$, in \eqref{eq:Shunsuke_RR_rec}, for $k=0,1,\dots$: \begin{align*} c^{(1)}_k\, &: 1,q,0,q^4,-q^7,q^9+q^{11},-q^{13}-q^{14}-q^{16},\dots,\\
 c^{(2)}_k\, &: 1,0,q^2,-q^4,q^6+q^7,-q^9-q^{10}-q^{11},\dots,
\end{align*} we observe that the minimal $q$-exponent in this sequence does not follow a simple quadratic polynomial, but the interlacing subsequences $c^{(i)}_{2k}$ and $c^{(i)}_{2k+1}$ do. It is straightforward to find the recurrences for these subsequences from 
\eqref{eq:Shunsuke_RR_rec}:
\begin{align}
c^{(i)}_{2k} &= q^{2k}(1+q+q^{2k-3})c^{(i)}_{2k-2}-q^{4k-1} c^{(i)}_{2k-4},\\
c^{(i)}_{2k+1} &=q^{2k+1} (1+q+q^{2k-2})c^{(i)}_{2k-1}-q^{4k+1} c^{(i)}_{2k-3},
\end{align} 
Similarly, it is easy to do formal substitutions and find the recurrences satisfied by the new sequences; this can be done automatically by \cite{ablinger2021qfunctions} too. These are
\begin{align} 
c^{(i)}_{2k} &= q^{k^2+k} c'^{(i)}_{k}\text{ and }c^{(i)}_{2k+1} = q^{(k+1)^2} c''^{(i)}_{k},
\intertext{where}
c'^{(i)}_k &= (1+q+q^{2k-3})c'^{(i)}_{k-1}-q c'^{(i)}_{k-2},\\
c''^{(i)}_k &= (1+q+q^{2k-2})c''^{(i)}_{k-1}-q c''^{(i)}_{k-2},
\end{align} 
are satisfied with the initial values $c'^{(1)}_{0}=c''^{(1)}_{0}=c''^{(1)}_{1}=c'^{(2)}_{0}=c'^{(2)}_{1}=1$, $c'^{(1)}_{1}=c''^{(2)}_{0}=0$, and $c''^{(2)}_{0}=-1$. We suppose that \[c'^{(i)}_k = \sum_{i=0}^k \hat{c}^{(i)}_l \qbinom{k}{l}\] then factorial basis method finds and proves a 4th order annihilator for $\hat{c}^{(i)}_l$: 
\begin{align*}
S_l^4 &- (1+q-q^{l+2}-q^{l+3}+q^{2 l+5})S_l^3 \\&+ (q-q^{l+2}-q^{l+3}+q^{2 l+3}+2 q^{2 l+4}-q^{3 l+5}-q^{3 l+6}-q^{3 l+7})S_l^2\\ &- (q^{2 l+2}-q^{3 l+3}-2 q^{3 l+4}-q^{3 l+5}+q^{4 l+5}+q^{4 l+6}+q^{4 l+7})S_l \\&-q^{3 l+2}+q^{4 l+3}+q^{4 l+4}-q^{5 l+5}.
\end{align*}
We can guess and prove (utilizing the already proven annihilator above) for $\hat{c}^{(i)}_l$ using \cite{ablinger2021qfunctions}. This annihilator, $S_l^2 -   q^{2l+1}S_l - q^{3l+1}$, applied to the sequences shows that \[\hat{c}^{(i)}_{l+2} =  q^{2l+1} \hat{c}^{(i)}_{l+1}  + q^{3l+1} \hat{c}^{(i)}_{l}\] with the initial values $\hat{c}^{(1)}_0=\hat{c}^{(2)}_0=1$, $\hat{c}^{(1)}_1=-1$, and  $\hat{c}^{(2)}_1=0$, are satisfied. We now make the substitution $\hat{c}^{(i)}_l = q^{-(l-1)^2} \bar{c}^{(i)}_l$. The new sequence $ \bar{c}^{(i)}_l$ satisfies the recurrence  \[\bar{c}^{(i)}_{l+2} =   \bar{c}^{(i)}_{l+1}  + q^{-l+1} \bar{c}^{(i)}_{l},\] with $\bar{c}^{(1)}_{0}=-1/q$,  $\bar{c}^{(1)}_{1}=1$, $\bar{c}^{(1)}_{2}=\bar{c}^{(2)}_{0}=0$, $\bar{c}^{(1)}_{3}=1$,  and $\bar{c}^{(2)}_{1}=-1$. Therefore, we can use Cigler's formula~\eqref{eq:Cigler_formula} for $\bar{c}^{(i)}_{l}$ when $l\geq 2$ for $i=1$ and $l\geq 0$ when $i=2$ with suitable substitutions. Back tracking substitutions yield the formulas for $c^{(i)}_{2k}$ as in \eqref{eq:2k_1} and \eqref{eq:2k_2}. Formulas for \eqref{eq:2kp1_1} and \eqref{eq:2kp1_2} are constructed through analogous steps.\qed
\end{example}

Staying within the theme of Rogers--Ramanujan identities, we give a small example of the product basis technique. 


\begin{example}\label{ex:AvEH} We can attempt to find a formula for the recurrence \[
e_{n+2} =(1 + q^{n+1} t - q^{2 n+2} t - q^{2 n+3} t) e_{n+1} + q^{3n+2} (1 - q^{n+1}) t^2 e_n, \] with initial conditions $e_0=1$ and $e_1=1-q t$. After exploring some other bases and failing, we can try to expand the sequence $e_n$ as \[e_n = \sum_{k=0}^n e'_k\, q^{nk} \qbinom{n}{k}.\] 
For doing so, we compute the product basis of $\cP^{(1)}$ and $\cC(1,0;0;1)$ (see Example~\ref{exm:power} and Lemma~\ref{lem:qbinom_comp_E}). The construction of this basis follows the steps of Example~\ref{Example_Product}. After constructing the compatibilities of this basis, we split the basis into 2 sections and analyze the even section. Then, constructing the recurrence of the summand $e'_k$ we see that $e'_k$ satisfies the recurrence $e'_{k+1} =- t e'_k$ and has the initial value $e'_0=1$. This is enough to deduce that $e'_k= (-t)^k$.
\end{example}

These types of series and recurrences are common in $q$-series and partitions \cite{uncu2021double, andrews2023singular, berkovich2019polynomial, bridges2022weighted, uncu2020polynomial}. In fact, we constructed the simple example, Example~\ref{ex:AvEH}, by looking at a recent paper of Andrews--van Ekeren--Heluani \cite[(3.4.1)]{andrews2023singular}. In their example, there is also a quadratic term: $q^{k^2}$ in the series. That extra factor only makes $e_n$ satisfy a higher order recurrence and makes it unnecessarily unpleasant to present here. However, the technique and its application work exactly the same.

Finally, we quickly give the outline of the proof of Theorem~\ref{thm:Main_Shunsuke_thm2}.
\begin{proof}[Proof of Theorem~\ref{thm:Main_Shunsuke_thm2}] 
The formulas for $c^{(i)}_{2k+\nu}$ in \eqref{eq:C_explicit} are made up of either two or four definite bilateral sums with the bounding variable $k$. For fixed $i=1$ or $2$ and $\nu=0$ or $1$, one can automatically find recurrences of these individual definite sums with respect to $k$ using an implementation of the Zeilberger's algorithm, e.g., see \cite{koutschanPhD, riese2003qmultisum, schneider2004summation}. Then one can find the recurrence satisfied by the sum of the definite sums using holonomic closure properties; this can be done by one of the implementations \cite{koutschanPhD, kauers2009mathematica}. Once this is done, it is easy to check the greatest common divisor (gcd) of this recurrence against the recurrences satisfied by $c^{(i)}_{2k+\nu}$ which can easily be deduced from~\eqref{eq:Shunsuke_RR_rec}. The gcd calculation can be done by \cite{ablinger2021qfunctions,koutschanPhD, kauers2009mathematica}. Once we establish that the gcd is the second order recurrence of $c^{(i)}_{2k+\nu}$ deduced from~\eqref{eq:Shunsuke_RR_rec}, all that remains is to check that the initial conditions match. This finishes the proof. 
\end{proof}

In Theorem~\ref{thm:Main_Shunsuke_thm}'s proof, we focus on finding and proving a closed formula of $c^{(i)}_{2k+\nu}$ using the factorial basis method. We note that, once the formula is found, we can prove Theorem~\ref{thm:Main_Shunsuke_thm} by following the steps outlined in the proof above. However, the factorial basis method, once found, also proves the recurrence of the summand. Hence, once a formula is reached it is already proven.

\section{Conclusions}\label{sec:conclusions}

In this paper, the factorial basis method has been extended into the realm of $q$-sequences. The factorial basis method allows proving well-known identities in $q$-combinatorics, but when selecting 
an appropriate factorial basis as a kernel, it is also effective in finding new identities or even showing such types of identities do not exist. A prototype for working with these $\beta(n)$-factorial basis automatically has been included in the SageMath package \texttt{pseries\_basis}, where the previous work in~\cite{JimenezPastor2023} was implemented. A SageMath notebook and a Mathematica file accompanying the examples of this paper can be found in the Git repository of the \texttt{pseries\_basis} package inside the notebooks directory.

We plan to extend and relax the compatibility relations explained in Section~\ref{sec:compatibility}. One possible relaxation is to include some $q^n$ shifts in $\alpha_{r,i}(m)$ in \eqref{equ:compatible}. We would also like to adapt the theory of shuffled basis introduced in \cite{JimenezPastor2023} to increase the diversity of the bases that one can try while attempting mathematical problems. Another task is to include negative shifts of power basis into this frame. This would allow bases of the form $q^{-n k }\qbinom{n}{k}$. 

On the implementation side, we plan to simplify the initiation and usage of the SageMath package \texttt{pseries\_basis}~\cite{JimenezPastor2023}. This implementation is aimed to be used on the side of the \texttt{qFunctions} Mathematica package~\cite{ablinger2021qfunctions}. Therefore, we plan to implement any necessary functions that would ensure the smooth transfer of the data from one system to the other.\\

\noindent{\bf Acknowledgements.} The first author was partially supported by the Poul Due Jensen Grant 883901. The second author would like to thank the EPSRC grant number EP/T015713/1 and the FWF grant P-34501N for partially supporting his research.

\bibliographystyle{ACM-Reference-Format}
\bibliography{biblio}

\end{document}